\documentclass{article}
\usepackage{tda}

\begin{document}

\title{Too Many or Too Few? Sampling Bounds for Topological Descriptors}

\author{
  Brittany Terese Fasy\thanks{Montana State University, Bozeman, MT, USA} \and
  Maksym Makarchuk\footnotemark[1] \and
  Samuel Micka\thanks{Western Colorado University, Gunnison, CO, USA} \and
  David L.~Millman\thanks{Blocky, Bozeman, MT, USA}
}

\maketitle

\section{Abstract}\label{sec:abstract}
Topological descriptors, such as the Euler characteristic function 
and the persistence diagram, have grown increasingly popular for
representing complex data. Recent work showed that a carefully
chosen set of these descriptors encodes all of the geometric and
topological information about a shape in~$\R^d$.
In practice, $\epsilon$-nets are often used to find samples in one of two extremes. On
one hand, making strong geometric assumptions about the shape allows us
to choose~$\epsilon$ small enough
(corresponding to a
high-enough density sample) in order to guarantee a faithful representation,
resulting in over-sampling. On the other hand, if we choose a larger~$\epsilon$ in order to allow
faster computations, this leads to an incomplete description of the 
shape and a discretized transform that lacks
theoretical guarantees. 
In this work, we investigate how many directions are really needed to represent
geometric simplicial complexes, exploring both synthetic and real-world
datasets.  We provide constructive proofs that help establish size bounds and 
an experimental investigation giving insights into the
consequences of over- and under-sampling.

\section{Introduction}\label{sec:intro}
Topological data analysis (TDA) provides tools
for extracting complex features from complex data.
Representing shapes using sets of topological descriptors, such as the
persistence diagram and the Euler characteristic function,
has drawn substantial interest in recent years; see~\cite{turner2014persistent,
curry2018many, ghrist2018euler, belton2018learning, fasy2018euler,
crawford2020predicting,betthauser,oudot2018inverse}.
More generally, several projects have found success in using topological
descriptors as features for machine
learning algorithms and statistical frameworks;
see~\cite{crawford2020predicting, hofer2019learning, qaiser2019fast,
lawson2019persistent,krishnapriyan2020robust,carriere2019general,
rostami2019deep}.

Much of the success of applying techniques from TDA comes from
the properties of the transformations that take
data and transform it into a set of topological descriptors.
One fundamental question is \emph{do these topological transforms help 
accurately capture the data}?~\cite{oudot2018inverse} argue that, because TDA techniques
are theoretically grounded,
there is a strong 
relationship between the representation of a shape
and the shape itself.

In this paper, we study the topological descriptors known as
the persistent homology transform~(\pht) and the Euler characteristic
transform~(\ect), and finite representations of them. The \pht\ (or \ect) is a function
that maps a simplicial complex in $\R^{d+1}$ to the family of persistence diagrams~(Euler
characteristic functions, respectively) parameterized by the set of
directions~$\S^{d}$.
\cite{turner2014persistent} showed that this map is injective for simplicial
complexes in $\R^2$ and~$\R^3$.

We turn to finite representations or discretizations of the PHT by
parameterizing the family of descriptors by $\Delta \subseteq \S^{d}$ instead of
$\S^d$ itself.
Thus, another natural question is: \emph{how small can $\Delta$ be}?
That is, how many directional persistence diagrams or Euler characteristic functions
are necessary to fully represent a shape?
\cite{fasy2019reconstructing} established an upper bound on the size of the set,
which is exponential in the dimension of the simplicial complex,
with a slight improvement for graphs embedded in~$\R^d$.
In \secref{lower-bounds}, we investigate lower bounds for the size of the
descriptor set of several different~transforms.

Recent results show that there are finite, faithful representations of
topological transforms 
when assumptions are made on the curvature of the
underlying shape~\cite{curry2018many}, when using the verbose
descriptors of simplicial complexes
in~$\R^d$~\cite{belton2018learning,belton2019reconstructing,fasy2019reconstructing},
or when 
using grayscale images and cubical complexes~\cite{betthauser}.
However, the resulting sets are often too large to be used in
practice, as they are exponential in dimension (unless additional assumptions are made).
In \secref{experiments}, we investigate how large these sets become in
practice on synthetic and real-world data sets.

Alternatively, one could use a coarse discretization of the sphere to create
a finite representation.
For example, \cite{hofer2019learning} found success on shape recognition tasks
by selecting 16 directions from~$\S^1$.
For each direction $s$, they computed one persistence diagram of a
height or lower-star filtration in direction~$s$.
Unfortunately, as we explore in \secref{experiments},
the simple discretization schemes would not necessarily capture
all of the geometry of the simplicial complex.
In \secref{loss}, we provide bounds on the error of the discretization in terms
of the Hausdorff distance.

\section{Background and Preliminary Results}\label{sec:background}
We explore topological descriptors that are created from
lower-star filtrations of (finite) simplicial complexes. In this section, we give a brief
background. We defer to~\cite[Ch.~1]{munkres} for a rigorous
introduction to homology and to~\cite[Ch.~VII]{edelsbrunner2010computational}
for an introduction to filtrations
and the computation of persistent~homology.

\paragraph{The Data: Embedded Simplicial Complexes}
Throughout this paper, we assume~$d \in \N := \{0, 1, \ldots  \}$ and we consider
geometric simplicial
complexes in~$\R^{d+1}$ (i.e., simplicial complexes embedded in~$\R^{d+1}$). Letting~$\simComp$ be a simplicial complex, we
use $\simComp_k$ to
denote the set of~$k$-simplices in~$\simComp$.
We call elements of $\simComp_0$ vertices and
$\simComp_1$ edges.
We assume that the simplicial complexes are not
geometrically degenerate.
Because we refer to this assumption often, we state it formally:

\begin{assumption}[General Position]\label{ass:general}
    Let $V \subset \R^{d+1}$ be a finite set. We say that $V$ is in
    \emph{general position} if no three points are collinear and no two points
    share a coordinate value. We say a geometric simplicial complex is in
    general position if its vertices are in general position.
\end{assumption}

\paragraph{The Descriptors: Invariants of Filtrations}
Simplicial complexes themselves are difficult to compare~\cite{ricci2018same}.
Instead, researchers turn to using descriptors of the data for summarization and
comparison. We focus on descriptors that summarize changes in
homology of a filtered topological space. A \emph{filtration} of a geometric simplicial
complex~$\simComp$ in $\R^{d+1}$
is a sequence of subcomplexes starting at the empty set and ending
with~$\simComp$; in particular, given a direction~$\dir \in \S^d$ and threshold $t
\in [0,\infty]$, the lower-level set of $\simComp$ with respect to~$\dir$ at
threshold~$t$~is:
\begin{equation}
    \lls{\simComp}{\dir}{t} := \{ \sigma \in \simComp ~|~
            \max_{v \in \sigma} \dir \cdot v \leq t\}
\end{equation}
We often refer to the threshold~$t$ as the \emph{height} parameter.
Noticing that, for all $t \leq t'$, the lower-level sets nest $\lls{\simComp}{\dir}{t} \subseteq
\lls{\simComp}{\dir}{t'}$,
the \emph{lower-star filtration} of $K$ with respect to direction $s$
is the following parameterized family of simplicial complexes:
\begin{equation}
    \filt{K}{\dir} := \{ \lls{\simComp}{\dir}{t} \}_{t \in \R}.
\end{equation}
We say that two directional filtrations, $\filt{K}{\dir}$ and $\filt{K'}{\dir'}$,
are equivalent if $K = K'$ and $\dir$ and $\dir'$ see the vertices in the same
order (but perhaps at different heights);
we denote this as $\filt{K}{\dir} \cong \filt{K'}{\dir'}$.

For a lower-star filtration,
we observe that topological changes only happen at the heights
of vertices (i.e., when $t = s \cdot v$ for some $v \in \simComp_0$).
However, while every topological change happens at the height of a vertex,
not every vertex witnesses a topological change.
For example, consider the simplicial complex~$K$ shown in
\figref{construction-triangle}, and let~$\dir = e^{i\pi/2}$ (the
$y$-direction). The vertices are seen in the following order:~$v_1$, then $v_2$,
then~$v_3$.
For~$\epsilon$ small enough,~$\lls{\simComp}{\dir}{\dir\cdot v_2 + \epsilon} =
\{v_1, v_2, [v_1, v_2]\}$
deformation retracts onto~$\lls{\simComp}{\dir}{\dir\cdot v_2 - \epsilon} = \{ v_1
\}$, which means that the map~$\lls{\simComp}{\dir}{\dir\cdot v_2 - \epsilon}
\hookrightarrow \lls{\simComp}{\dir}{\dir\cdot v_2 + \epsilon}$ induces an isomorphism
on homology (i.e., there is no topological change at the height of~$v_2$).

\begin{figure}[htb]
    \centering
    \includegraphics[height=0.5in]{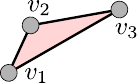}
    \caption{Example of a simplicial complex with three vertices, three edges,
        and one triangle. In the $y$-direction, the corresponding filtration
        sees three distinct simplicial complexes:~$\{v_1\} \subset  \{v_1,v_2,
        [v_1,v_2]\} \subset
        \{v_1,v_2,v_3,[v_1,v_2],[v_1,v_3],[v_2,v_3],[v_1,v_2,v_3]\}$. However,
        the lower-star filtration in the $y$-direction only sees one topological
        change, when a
        a new connected component is introduce at the height of $v_1$.
        }
    \label{fig:construction-triangle}
\end{figure}

We define a \myemph{topological descriptor} of a filtration to be
any topological invariant of the
filtration.
Examples of topological descriptor types include persistence
diagrams~\cite[Ch.~7]{edelsbrunner2000topological},
persistence landscapes~\cite{bubenik2015statistical},
Betti functions~\cite{van2011alpha}, and
Euler characteristic functions~\cite{torres,vogeley1994topological}.
The \emph{persistence diagram} is a graded multiset of
points in the extended plane, where each point is graded, or labeled, with a dimension.
The points correspond to homological features (colloquially described as
connected components, tunnels, and voids in low dimensions) that are present
for an interval of heights. The endpoints of that interval are the coordinates
of the point in the persistence diagram.
The persistence diagram is a complete, discrete invariant of a filtration, up to
homology. In particular, from the persistence diagram
alone, at
each height~$t$, we know the homology of $\lls{\simComp}{\dir}{t}$.
Naturally, functional summaries of persistence
diagrams are also topological descriptors, and provide access to the tools of
functional data analysis~\cite{crawford2020predicting,berry2018functional}.
For example, the persistence landscape (defined
in~\cite{bubenik2015statistical}) is a function $L \colon \R \times \Z \to
\R$ and is in one-to-one correspondence with persistence diagrams, but has the
advantage that means are well-defined, and, in fact, can be taken pointwise
(however, those means might not map back
to a persistence diagram).
For each $i \in \N$, the \emph{Betti function}~$\beta_i \colon \R \to \R$ assigns, for each~$t \in
\R$, the rank of the~$i$-th homology group,~$H_i(\lls{\simComp}{\dir}{t})$.
The \emph{Euler characteristic function}~$\chi \colon \R \to \R$ maps each
height to the Euler characteristic of the corresponding lower-level
set:~$\chi(t) = \sum_{i=0} (-1)^{i}\beta_i(t)$.
The Euler characteristic function and a variant of it called the smooth Euler
characteristic function are used by diverse groups of researchers today;
see~\cite{crawford2020predicting,bobrowski2017vanishing}.

Topological descriptors are invariants of the \emph{filtration}, not of the
\emph{underlying shape} itself.  For example, chirality is not captured by a single
descriptor.  However, given a descriptor in all directions,
the infinite set of descriptors~(labeled with their directions) contains more
information.  A \myemph{directional topological transform}
of type $\type$ (simply, a $\type$-transform in this paper) of a simplicial
complex $\simComp$ is the following
family of topological descriptors:
    \begin{equation}
        \fulltrans^{\type}(\simComp) := \{ (\dir,\type_K(\dir)) ~|~
        \dir \in \S^{d} \},
    \end{equation}
where $\type_K(\dir)$ is the topological descriptor of $\simComp$ of type $\type$ in direction
$\dir$.

For many common topological descriptor types,
the directional topological transform is a complete
invariant of the shape.
We call this set of descriptor types the \myemph{fundamental
descriptor types}; in particular, we note that
the
persistence diagram and the Euler characteristic function are fundamental
descriptor types.

While directional topological transforms use an infinite set
of descriptors, a finite set of well-chosen directions often suffices to represent the
transform.

\begin{definition}[Discrete Topological Transform]\label{def:discrete}
    Let $\type$ be a topological descriptor type and let $\Delta \subset \S^{d}$
    be a discrete set. We call $\Delta$ a set of directions.
    Then, we define~$\trans{\Delta}^{\type}$ as the function
    from the space of geometric simplicial complexes in $\R^{d+1}$
    to a space of topological descriptors of type $\type$ parameterized by~$\Delta$:
    \begin{equation}
        \trans{\Delta}^{\type}(\simComp) := \{ (\dir,\type_K(\dir)) ~|~
        \dir \in \Delta \}.
    \end{equation}
    Equivalently, $\trans{\Delta}^{\type}(K)$ is a set of direction /
    topological descriptor pairs.
    We call $\trans{\Delta}^{\type}(\simComp)$ \emph{faithful} if it contains enough
    information to recover $K$.
\end{definition}

In order to quantify how much of a simplicial complex is represented by a set of
topological descriptors,
we define \emph{observable} and~\emph{$\theta$-observable},
generalizing~\cite[Defs.~7.2 and~7.3]{curry2018many}.\footnote{We generalized
Euler observable from \cite{curry2018many} to simply observable by any topological
descriptor.
Letting $\delta=\sqrt{2-2\cos{\frac{\theta}{2}}}$, what we define as~$\theta$-observable
implies `at least'~$\delta$-observable, as defined in~\cite{curry2018many}.}

\begin{definition}[Observable]
\label{def:eul-obs}
    Let~$\simComp$ be a geometric simplicial complex in~$\R^{d+1}$. Let~$v \in
    \simComp_0$ and $\dir \in \sph^d$.
    Then, we say that $v$ is \emph{observable} from~$\dir$
    if there is a topological change at at $v$ in direction $\dir$.
\end{definition}

\begin{definition}[$\theta$-Observable]
\label{def:theta-obs}
    Let~$\simComp$ be a geometric simplicial complex in~$\R^{d+1}$ and $v \in
    \simComp_0$.
    We say that $v$ is \emph{$\theta$-observable} if there exists a direction $\dir_0 \in
    \S^d$ such that $v$ is observable from~$\dir_0$ and for all $\dir$ such
    that $\cos^{-1}(\dir_0 \cdot \dir) \leq \theta$ (i.e., the angle
    between $\dir_0$ and
    $\dir$ is at most~$\theta$).
\end{definition}

For a vertex $v \in K_0$, the largest value of $\theta$ such that $v$ is
$\theta$-observable relates the ``flatness'' of the
simplicial complex at $v$.
To represent a simplicial complex with a set of topological descriptors,
a necessary condition is
that each vertex is observed.  Hence, we~define:

\begin{definition}[Observing Region]\label{def:obs-stratum}
    Let~$\simComp$ be a geometric simplicial complex in~$\R^{d+1}$ and $v \in
    \simComp_0$.
    The \emph{observing region} of~$v$~is:
    \begin{equation}
        \obs{v}{\simComp}: = \{ \dir \in \sph^{d} ~|~
        \text{ $v$ is observable from $\dir$} \}.
    \end{equation}
\end{definition}

The \emph{size} of the observing region is its Lebesgue
measure~\cite{lebesgue-phd}.
Furthermore, the observing regions are related to a stratification of $\sph^{d}$.

\begin{definition}[Coarse Stratification]\label{def:orth-strat}
    Let~$\simComp$ be a geometric simplicial complex in~$\R^{d+1}$.
    The \emph{coarse stratification} of $\S^d$ is the partition
    induced by the equivalence relation $\simeq$, where, for all $\dir_1,\dir_2 \in
    \S^d$ in the same strata,
    \begin{equation}
       \dir_1 \simeq \dir_2 \iff \ord{\simComp_0}{\dir_1}=\ord{\simComp_0}{\dir_2},
    \end{equation}
    where $\ord{\simComp}{\dir}$ denotes the partial order of vertices in $K_0$
    with respect to direction $\dir$.
\end{definition}

By this definition, knowing the vertex
set is sufficient for knowing the coarse stratification.
Furthermore, we note that this definition uses only a partial order for an
arbitrary direction (as opposed to a total order),
as there are directions where two or more vertices are at the same height.
The cells
whose partial order is a total order are the $d$-dimensional~cells.

\begin{corollary}[Relation Between $K$ and $\theta$]\label{cor:relation}
    Let~$\simComp$ be a geometric simplicial complex in $\R^{d+1}$.
    The diameter of the smallest~$d$-dimensional stratum in the coarse stratification of
    $\simComp$ is equal to the smallest angle between vectors with endpoints in the vertex set.
    Furthermore, if $\theta := \min\{ \angle u,v,w ~|~ uvw \in \simComp_0\} $,
    then all vertices in $K_0$ are~$\theta$-observable.
\end{corollary}

In~\cite[Theorem 7.14]{curry2018many}, a representative persistence diagram (or Euler
characteristic function)
from each $d$-cell of the coarse
stratification, along with the coarse stratification itself, are used as a finite
representation of the PHT or ECT.
In fact, this sufficiency result holds for
all fundamental descriptor types with the same proof as in
\cite[Theorem~7.14]{curry2018many}, which allows us to use stratifications in order to
compare sets of simplicial complexes:

\begin{lemma}[Sufficient Faithful Set]\label{lem:suff-faithful}
    Let $K$ be a geometric simplicial complex in $\R^{d+1}$, $\type$ a
    fundamental descriptor type, and~$P\subseteq \S^{d}$ be a set that hits all the
    highest-dimensional strata of the coarse stratification.
    Then, $\trans{\Delta}^{\type}(\simComp)$ along
    with the
    coarse stratification is faithful.
\end{lemma}

Leveraging this lemma, we know that in order to distinguish different shapes,
it would suffice to select
a descriptor in each strata.

\begin{theorem}[Sufficient Comparison Set]\label{thm:sufficientcmpset}
    Let $\type$ be a fundamental descriptor type,
    with $\delta \colon \type \times \type \to \R$ a distance
    metric.\footnote{Here, we slightly abuse notation and let $\type$ denote
    both the descriptor type and the set of all descriptors of that type.}
    Let~$\gscSpace$ be a set of geometric simplicial complexes in~$\R^{d+1}$.
    Let~$P \subset \S^d$ be a finite set that hits every d-cell of the coarse
    stratifications for each $K \in \gscSpace$.
    Then, the following function is a distance
    metric:
    \begin{align}
        d_{\gscSpace} \colon \gscSpace \times \gscSpace &\to \R \\
         (\simComp,\simComp') &\mapsto \sum_{\dir \in P}
         \delta\left(
                    \type_{\simComp}(\dir),
                    \type_{\simComp'}(\dir)
                \right).
    \end{align}
\end{theorem}

\begin{proof}
    Non-negativity, symmetry, the triangle inequality, and the fact
    that, for all~$\simComp \in \gscSpace$,~$d_{\gscSpace}(\simComp,\simComp)=0$
    follow directly from
    $\delta$ being a distance metric.

    To show the final property (positivity), let $\simComp \neq \simComp' \in
    \gscSpace$ and
    suppose, by contradiction, that~$d_{\gscSpace}(\simComp,\simComp')= 0$.
    Because $P$ hits all of the~$d$-cells of the coarse stratification for $K$,
    by \lemref{suff-faithful}, there must exist some direction~$s
    \in P$ such that~$\filt{\simComp}{s} \not \cong
    \filt{\simComp'}{s}$.  Because $\delta$ is a metric, we then know
    that~$\delta( \type_{\simComp}(\dir), \type_{\simComp'}(\dir) ) > 0$.
    Thus, by the definition of $d_{\gscSpace}$,
    we conclude~$d_{\gscSpace} (\simComp,\simComp') > 0$.
\end{proof}

One way to apply \thmref{sufficientcmpset} is to define a set of directions for
a fixed database of shapes.
Let
$\gscSpace = \{ \simComp_1, \simComp_2, \ldots, \simComp_n \}$ be a set
of geometric simplicial complexes.
For~$\simComp_i \in \gscSpace$, let $\epsilon_i$ be the size of the minimal
strata in the coarse stratification of $\simComp_i$.
Let $\epsilon = \min_i \epsilon_i$ and let $\Delta$ be $\epsilon$-net over $\S^d$,
First, observe that for each~$\simComp_i \in \gscSpace$ and~$\type$ a fundamental
descriptor type, the transform
$\trans{\Delta}^{\type}(\simComp_i)$ produces a set of descriptors that uniquely
represents the shape.
Second, observe that~$\Delta$ satisfies the assumptions on the direction set of
\thmref{sufficientcmpset}, allowing us to compare our complexes.
Thus, understanding the minimal stratum size is particularly~relevant.

\section{Lower Bounds on Representation}\label{sec:lower-bounds}
In this section, we provide lower bounds on the number of
topological descriptors necessary to guarantee that the observing region of each
vertex is hit by at least one descriptor.
This lower bound grows at least linearly with the size of the vertex set, as
the observing region for a degree-two vertex
grows arbitrarily small as it
approaches the affine space spanned by its two adjacent vertices.

\begin{lemma}[Observing Regions Witness Local Max and Min]\label{lem:miss-deg-two}
    Let~$\simComp$ be a geometric simplicial complex in~$\R^{d+1}$.
    Let $v\in K_0$
    be a degree-two vertex with adjacent vertices $u$ and~$w$. Define the
    following two sets:
    \begin{align*}
        P_{1} &:= \{\dir \in \sph^{d}\ |\ \dir \cdot u < \dir \cdot v \text{ and }
                        \dir \cdot w < \dir \cdot v\},\\
        P_{2}  &:= \{\dir \in \sph^{d}\ |\ \dir \cdot u > \dir \cdot v  \text{ and }
                        \dir \cdot w > \dir \cdot v\}.
    \end{align*}
    Then, $\obs{v}{\simComp} = P_1 \cup P_2$.
\end{lemma}

\begin{proof}
    Let $\dir \in \obs{v}{\simComp}$.
    Because~$v$ is degree two, when
    looking at $v$ in direction $\dir$, we encounter one of three cases:
    (1) $v$ is a local maximum in direction $\dir$. This happens if and only if~$\dir \in
    P_{1}$;
    (2) $v$ is seen as a local minimum in direction $\dir$.  This
    happens if and only if~$\dir \in
    P_{2}$; (3) one edge goes into~$v$ from below and the other edge starts at
    $v$ and goes up with respect to direction $s$. In this case,~$v$ is not observable from
    $\dir$.  Hence, we conclude~$\obs{v}{\simComp} = P_1 \cup P_2$.
\end{proof}

\begin{figure}[htb]
    \centering
    \includegraphics[width=\textwidth]{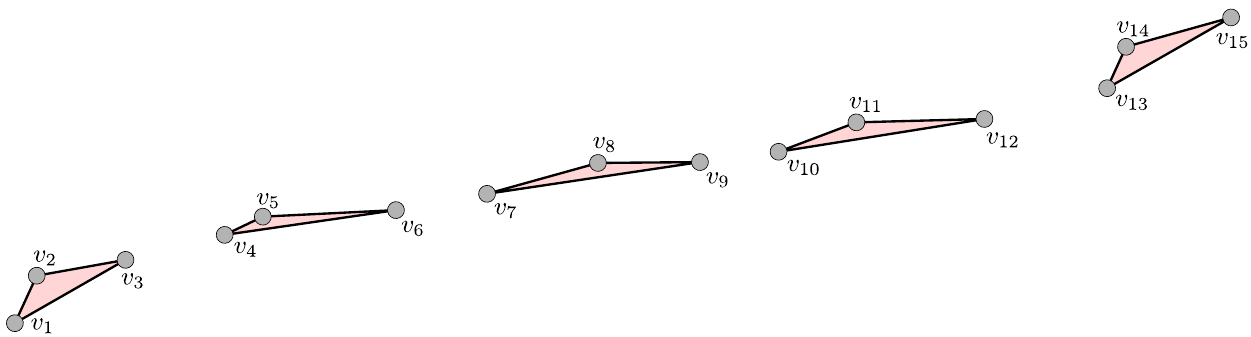}
    \caption{A simplicial complex with
        $15$ vertices.  For any fundamental descriptor type,
        at least five topological descriptors are needed to uniquely
        represent this simplicial complex; in general, there exist configurations
        of $n$ simplices that require $\Omega (n_0)$ topological descriptors for
        a faithful representation.
        }
    \label{fig:construction}
\end{figure}

Consider the simplicial complex in \figref{construction} comprised of five
triangles, along with their edges and vertices.  This complex has
the additional feature that no two edges are parallel.  As a result, we
imagine moving the vertices~$v_2,v_5,v_8,v_{11},$ and $v_{14}$ so close to the
complementary edge of their respective triangles
that the three edges defining the triangle are nearly
parallel.  We generalize this idea as follows:

\begin{theorem}[Linear Size Descriptor Set]\label{thm:lower:bc}
    There exist simplicial complexes that
    require at least a linear number of topological descriptors (with respect
    to the number of vertices) for unique~representation.
\end{theorem}

\begin{proof}
    We give an example of such a simplicial complex.
    Let $n \geq 3$ be a multiple of three.  Let $\theset = \{2, 5, 8, \ldots,
    3n-1 \}$.  For each $i \in \theset$, define the following two vertices:
    \begin{equation}
        v_{i-1} := (i-1,i-1)
         \text{ and }
        v_{i+1} := (i+1,i+1+\frac{1}{i}).
    \end{equation}
    Let $\delta=\frac{1}{3n}$, and let $\ell_i$ be the line that goes through
    $v_{i-1}$ and $v_{i+1}$.  Let~$i,j \in \theset$ such that $j \neq i$.  Then,
    by construction,
    the difference of slope between~$\ell_i$ and $\ell_j$ is at least $\delta$.
    By an application of the intermediate value theorem, we choose an
    $\epsilon > 0$ so that the point
    $v_i =  (i,i+\frac{1}{2i} + \epsilon)$ is chosen such that the slopes of lines through
    $v_{i-1}$ and $v_{i}$ and through $v_i$ and $v_{i+1}$ differ by at
    most~$\frac{\delta}{2}$.  Then, we define a simplicial complex as follows:
    \begin{align*}
        \simComp_0 &:= \{ [v_i] \}_{i \in \{1,2,\ldots,3n\}}\\
        \simComp_1 &:=  \{ [v_{i-1},v_i], [v_i,v_{i+1}], [v_{i-1},v_{i+1}] \}_{i
        \in \theset}\\
        \simComp_2 &:=  \{ [v_{i-1},v_i,v_{i+1}] \}_{i \in \theset}
    \end{align*}
    No higher dimensional simplices are in $\simComp$.

    In order to represent $\simComp$, each vertex needs to be observed. In
    particular,
    for each $i \in \theset$,
    we must use a topological descriptor from a direction sampled
    from~$\obs{v_i}{\simComp}$.
    Let $i,j \in \theset$.
    Then, by choice of $\epsilon$,
    $\obs{\simComp}{v_i} \cap \obs{\simComp}{v_j} = \emptyset$.
    In other words, each such $v_i$
    has a disjoint observing region, so we need at least~$\Omega(n)$
    topological descriptors to represent~$\simComp$.
\end{proof}

We note here that, if the $i$-th vertex is not observed, then $K$ cannot be
distinguished from the simplicial complex that removes the $i$-th vertex and its
cofaces (i.e., replaces the $i$-th triangle with
a line segment).

\section{Experimental Investigation}\label{sec:experiments}
In this section, we perform a series of experiments.
First, we use an~$\epsilon$-net for sampling to highlight the issue
of oversampling in geometry-based approaches. Then, we
examine another commonly used technique in the field: uniformly sampling a fixed number
of directions. Our goal is to analyze how many of the observing regions are
missed as
the number of vertices increases. At the end of this section, we demonstrate how the
theoretical lower bounds on representation work in practice.
All code is available in a public git
repo.\footnote{Code available at \url{https://github.com/compTAG/topo-descriptor-experiments}}

\subsection{Data}\label{sec:data}
In our experiments,
three planar data sets: random point
clouds, the MPEG7 dataset \cite{mpeg7}, and a subset of the EMNIST dataset
(which is an extension of the MNIST dataset)
\cite{cohen2017emnist}.
For the MPEG7 and MNIST datasets, we describe a simple but
reasonable set of preprocessing steps for a TDA pipeline based on contours,
similar to that found in~\cite{bai09,
belongie02, felzenszwalb07}.

\paragraph{\randpts}
The point cloud dataset, named \randpts, consists of varying point set sizes
in which the points are
drawn i.i.d.\ from the uniform distribution over the $[0,10]^2$ box.
In particular, for each $k$ in ${3, 5, 10, 20, \ldots, 100}$,
we generate 100 point clouds of size $k$.

\paragraph{MPEG7 and EMNIST}
The MPEG7 dataset contains $70$ classes of binary images of shapes, such as animals.
The EMNIST by class dataset contains~$62$ classes of grayscale images of
handwritten characters; we use the first $100$ images in each class.

For these images, we extract the boundary from the object in the
image such that the vertices conform to
\assumptref{general} and the boundary is a simple polygon.
We simplify the boundaries at two levels using tools from OpenCV~\cite{opencv}
and NetworkX~\cite{hagberg2008exploring}, which
results in a total of four sets of graphs named
\mpegI, \emnistI, \mpegV, and \emnistV, where the
subscript corresponds to the simplification parameter.
Specifically, first, we convert to a binary image based on a threshold.
In MPEG7, images are black and white, so we use a threshold of zero.
In EMNIST, we compute a global threshold and apply it to all images.
To compute the global threshold,
we find the optimal threshold of the first image in each class using
Otsu's algorithm~\cite{otsu}, then
we set the global threshold as the average of these thresholds.
The global threshold was~$102.951$ with a standard
deviation of $7.631$.
Second, for each binary image, we compute one contour curve at two
simplification levels as follows.
We first apply \cite{suzuki85} to extract a set of contours
and keep the contour with the longest arc length.
Afterwards, we simplify each contour using
Douglas–Peucker~\cite{douglas1973algorithms}.
Given a parameter $\epsilon > 0$,
Douglas–Peucker iteratively simplifies a curve by replacing
a subcurve with a line segment, provided that the change in Hausdorff distance
is smaller than~$\epsilon$.
For a contour with arc length $s$,
we produce the two simplifications by setting
$\epsilon \in { .001 s, .005 s}$.
Third, because the contours have vertices that lie on a grid, and
\assumptref{general} states that $x$- and $y$-coordinates are distinct,
we perturb each coordinate by a value chosen uniformly from~$[-.01,.01)$.
Fourth, we clean the data by dropping graphs that contain three collinear
vertices or are not a simplicial complex.
For example, when simplifying contours in the second step,
a vertex removal may cause the contour to self-overlap.
More sophisticated simplification algorithms, such as \cite{estkkowski01},
avoid overlaps.
We also removed erroneous data,
such as one MPEG7 image with inverted colors.
After generating and cleaning the data, we convert each contour to a graph
and build our output sets.
After the preprocessing and pruning,
\mpegI\ has $1387$ graphs,
\emnistI\ has $5563$ graphs,
\mpegV\ has $1364$ graphs,
and \emnistV\ has $5751$ graphs.

\subsection{Too Many: Geometric Based Discretization}
\label{ss:smallest_stratum_exp}
By \thmref{sufficientcmpset}, one way to build a faithful
descriptor set is to find a finite sample \(P\) such that every stratum in the
coarse stratification is “hit” by a
point in \(P\). If we know the size of the smallest stratum in a coarse
stratification, then an alternate approach is to use an~\(\epsilon\)-net.

In this experiment, we investigate the relationship between the size of the smallest one-stratum (measured in radians)
and the number of vertices \(n_0\) as it increases across the datasets \randpts, \mpegI, and \emnistI. We show that
creating such an \(\epsilon\)-net results in too many descriptors to be analyzed. We depict the results as a scatter plot
on a \(\log\)-\(\log\) scale with best-fit lines in \figref{min-stratum-exp-001-approx}. Specifically, we examine how
the size of the smallest one-stratum (\(\log(m)\)) changes with the logarithm of
the number of vertices~(\(\log(n_0)\)).

\begin{figure*}[ht]
    \begin{subfigure}[b]{0.32\textwidth}
        \includegraphics[width=\textwidth]{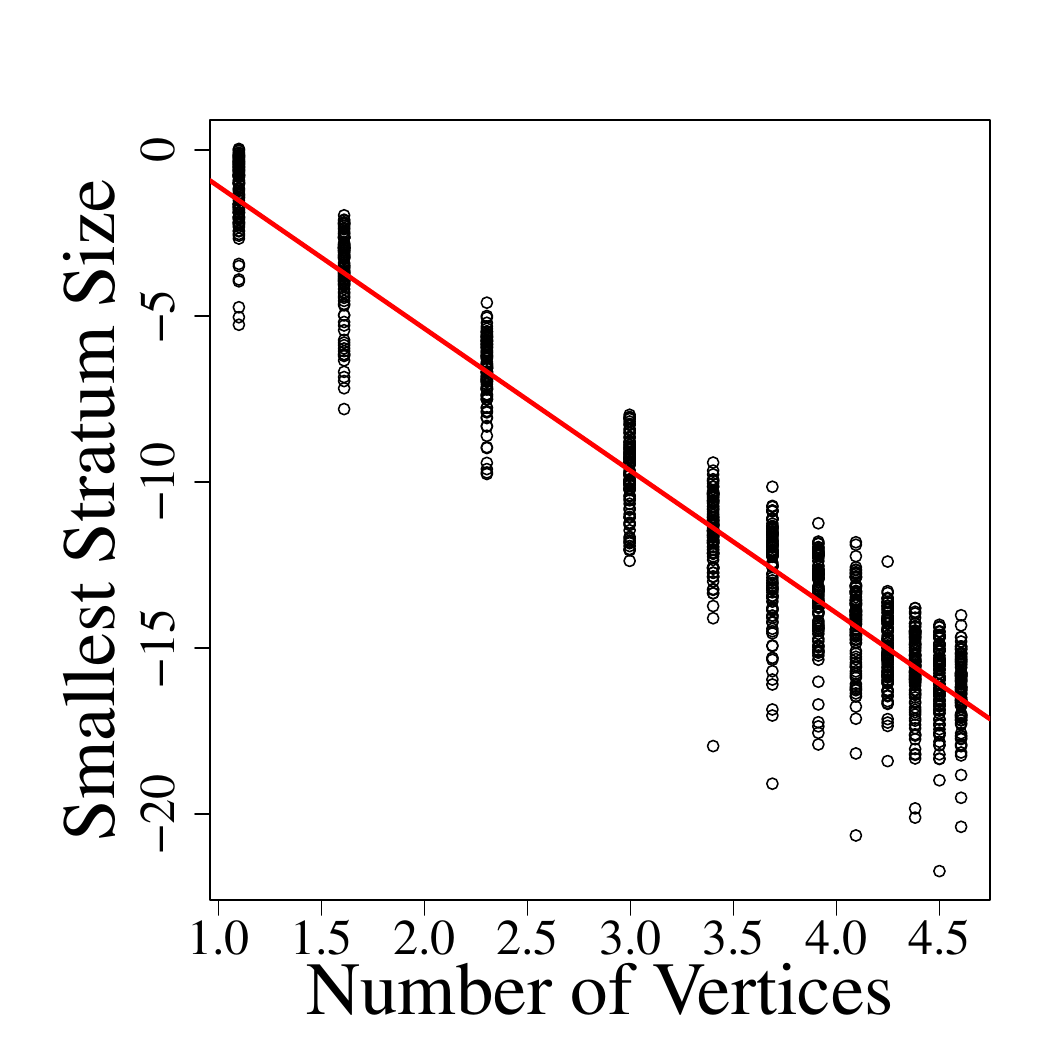}
        \caption{\randpts}
        \label{subfig:min-stratum-random-001-approx}
    \end{subfigure}
    \begin{subfigure}[b]{0.32\textwidth}
        \includegraphics[width=\textwidth]{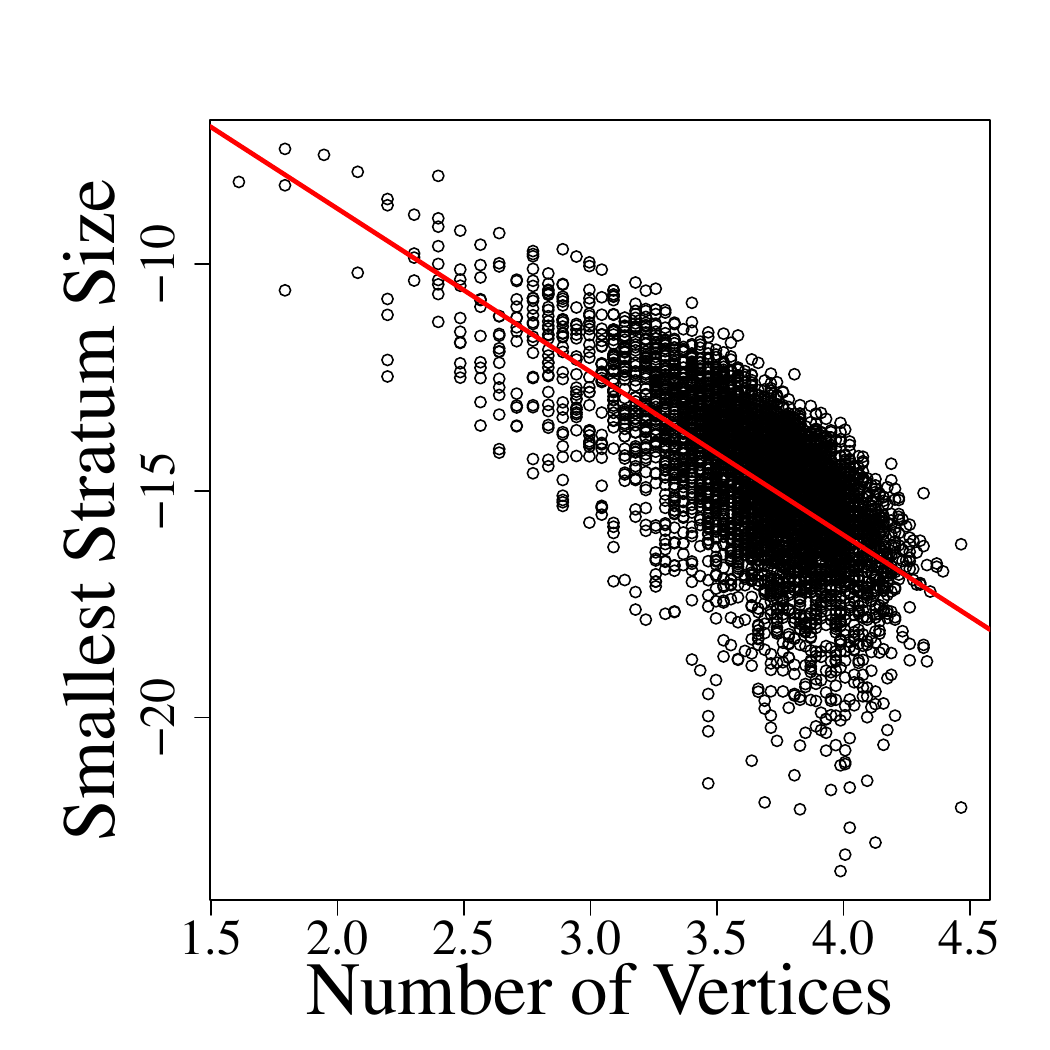}
        \caption{\emnistI}
        \label{subfig:min-stratum-mnist-001-approx}
    \end{subfigure}
    \begin{subfigure}[b]{0.32\textwidth}
        \includegraphics[width=\textwidth]{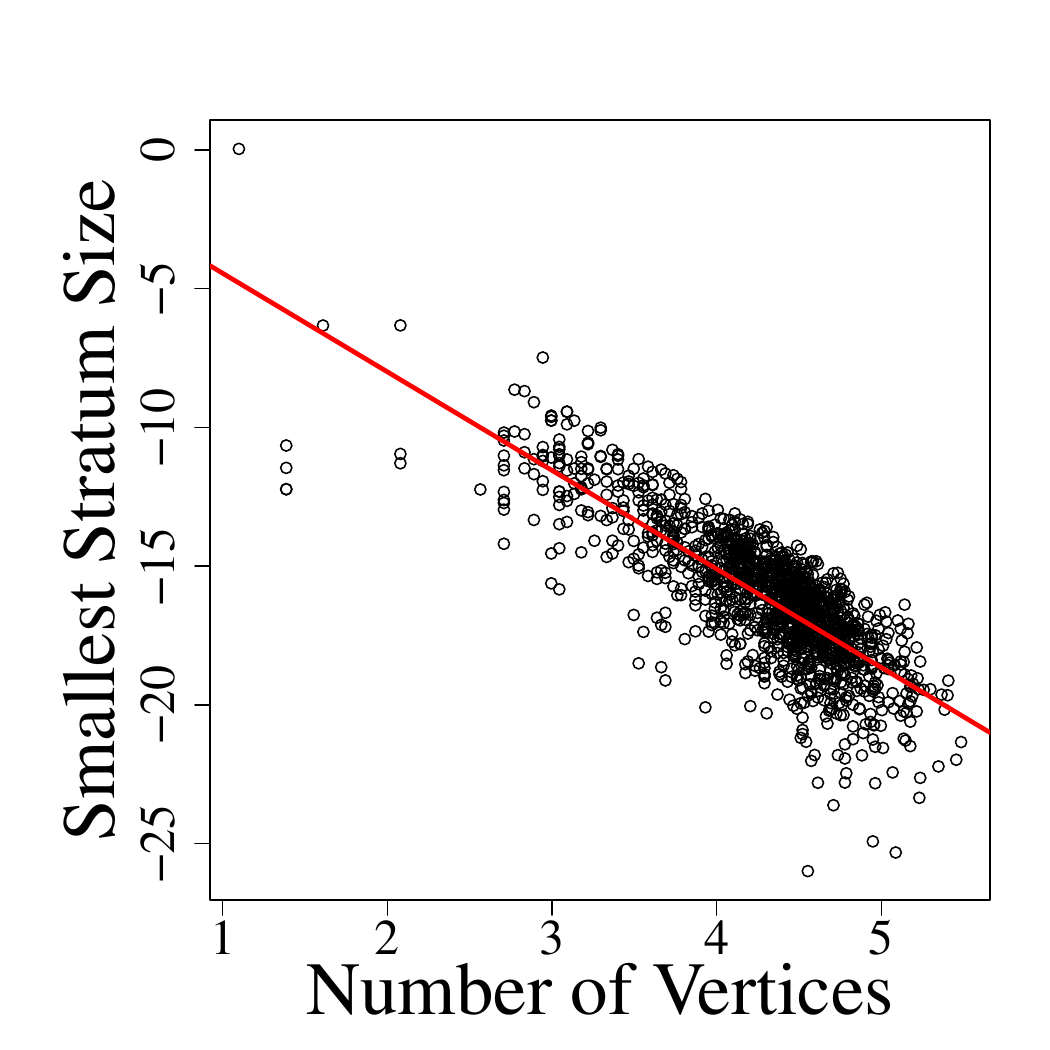}
        \caption{\mpegI}
        \label{subfig:min-stratum-mpeg7-001-approx}
    \end{subfigure}
    \caption{Log-log plot of the smallest stratum size versus the number of
        vertices for datasets
        \randpts, \emnistI, and \mpegI.
     }
     \label{fig:min-stratum-exp-001-approx}
\end{figure*}

Across all three datasets, the smallest stratum size decreases proportionally as the number
of vertices increases. Even for graphs with thousands of vertices, the minimum stratum size
drops below $10^{-5}$, making uniform discretization impractical due to the excessive number of
descriptors it would generate. The best-fit regression lines confirm a strong negative correlation
between the smallest one-stratum size and the number of vertices:

\begin{itemize}
    \item \randpts, $\log(m) = 3.17515 - 4.28122 \log(n_0)$.
    \item \emnistI, $\log(m) = -1.5887 - 3.5964 \log(n_0)$.
    \item \mpegI, $\log(m) = -0.89784 - 3.55102 \log(n_0)$.
\end{itemize}

In all three experiments,
as the number of vertices increases,
we see smaller
strata appearing with more drastic changes for smaller values of $n_0$.
The observation produces a particularly pessimistic view for trying to
discretize a transformation using the minimum observing region of the
dataset.
Indeed, we already see that for relatively small graphs with
$1000$’s of vertices, the minimum stratum size drops below $10^{-5}$,
which would produce descriptor sets containing hundreds of thousands
of descriptors, far too large to be of any use in practice.

\paragraph{Additional Considerations}

\begin{figure}
    \begin{center}
        \includegraphics[width=0.37\textwidth]{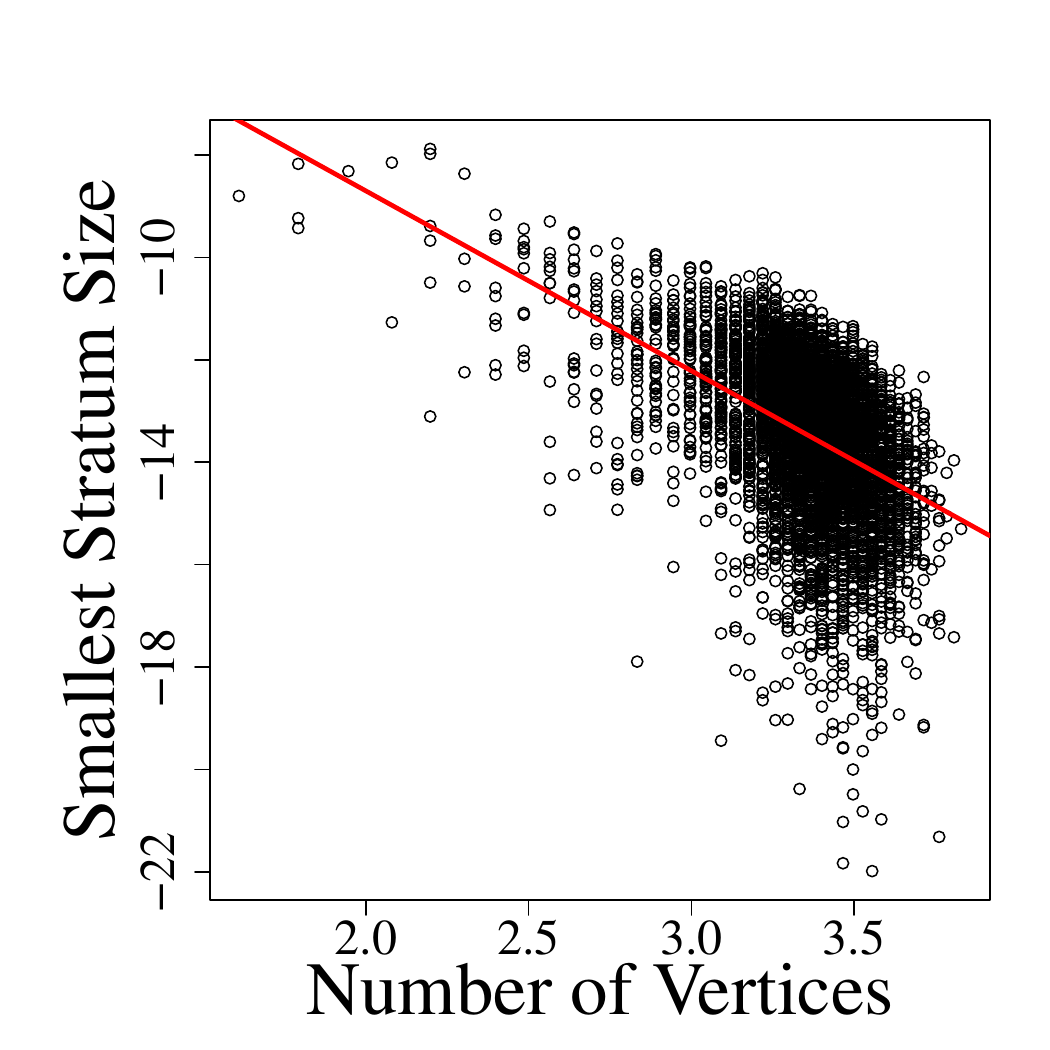}
        \label{subfig:min-stratum-mnist}
    \end{center}
    \caption{Log-log plot of the smallest stratum size versus the number of
    vertices for dataset \emnistV.}
    \label{fig:min-stratum-exp-005-approx}
\end{figure}

Because the contour approximations remove vertices,
it is natural to ask whether a coarser approximation
yields different dependencies between $m$ and $n_0$.
To test this, we consider the dataset \emnistV.
The resulting model still had a negative coefficient
for $\log(n_0)$, and the results are depicted in \figref{min-stratum-exp-005-approx}.
We note that, similar to the main experiments of this section,
a different coarsening does not substantially decrease the size of a
descriptor set created by uniformly discretizing based on the
minimal observing region.

\subsection{Too Few: Constant Size Discretization}
\label{ss:const_disc_exp}
We demonstrated that sampling directions
from every cell of the $\epsilon$-net yields a computationally
intensive volume of diagrams. Such an approach is impractical for
analyzing real-world data. Instead of sampling based on the
geometry of the data, it is common to use a fixed-size set of
directions. But how does it affect the results of the analysis?

To address this question, we conducted an experiment to examine how much
information about a simplicial complex can be captured using a fixed-size
direction set. Given a simplicial complex
$\simComp$, its coarse stratification~$S_\simComp$, and a
discretization $\Delta = \{ s_1, s_2, \ldots, s_n \} \subset \S^1$,
a stratum $k \in S_\simComp$ is considered \emph{hit} by the discretization
if $|\{s_i \in \Delta : s_i \in k \}| > 0$.

In this experiment, we fix a discretization of $\S^1$ and study
the proportion of strata from coarse stratifications that are hit
as we vary $n_0$. Specifically, we investigate how the number of captured
strata from the coarse stratification changes as the number of vertices within the simplicial complex varies.

For our experiments, we choose $2^{14} = 16384$ directions for the
discretization, which is much larger than is common in practice.
For example, \cite{turner2014persistent} uses a discretization of $64$
directions, and \cite{crawford2020predicting} use~$72$ directions.
We picked such a large set of directions to study the best-case scenario
for a constant discretization, as in practice, any reasonable discretization
would be coarser and miss a higher proportion of strata.

It is worth noting that for a graph with $n_0$ vertices, the number of strata
in the coarse stratification is $N = 2{n_0 \choose 2}$. To ensure that
there was ever a chance of hitting every stratum, we only include graphs
from our datasets whose coarse stratification contains fewer than $5000$
strata. The resulting sizes of the datasets were \randpts\ containing $900$ graphs,
\emnistI\ with $5543$ graphs, and \mpegI\ reduced to $521$ samples. By utilizing scatterplots, we visualize the relationship between
the number of hits and the number of vertices in the simplicial complex.

\begin{figure*}[h]
    \begin{subfigure}[b]{0.32\textwidth}
        \includegraphics[width=\textwidth]{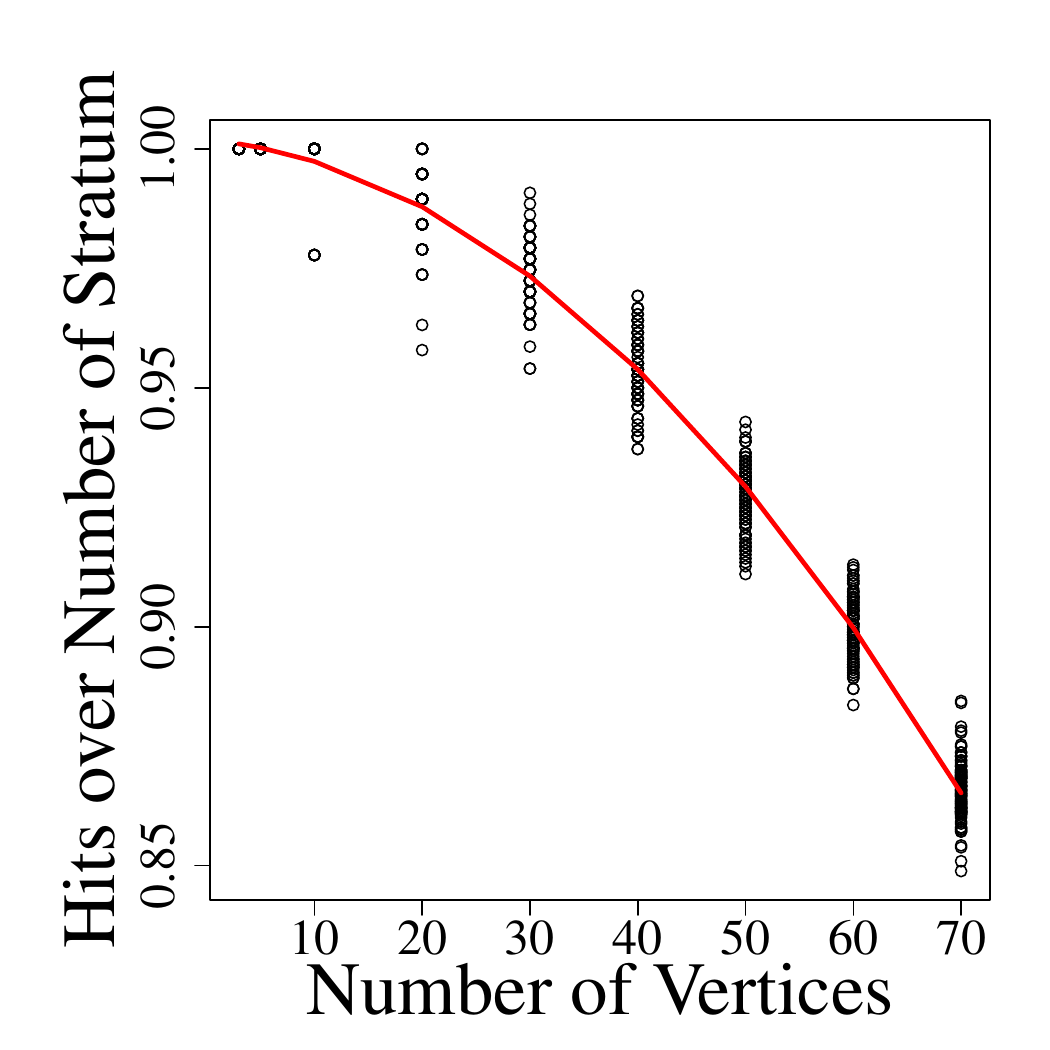}
        \caption{\randpts}
        \label{subfig:uniform-random}
    \end{subfigure}
    \begin{subfigure}[b]{0.32\textwidth}
        \includegraphics[width=\textwidth]{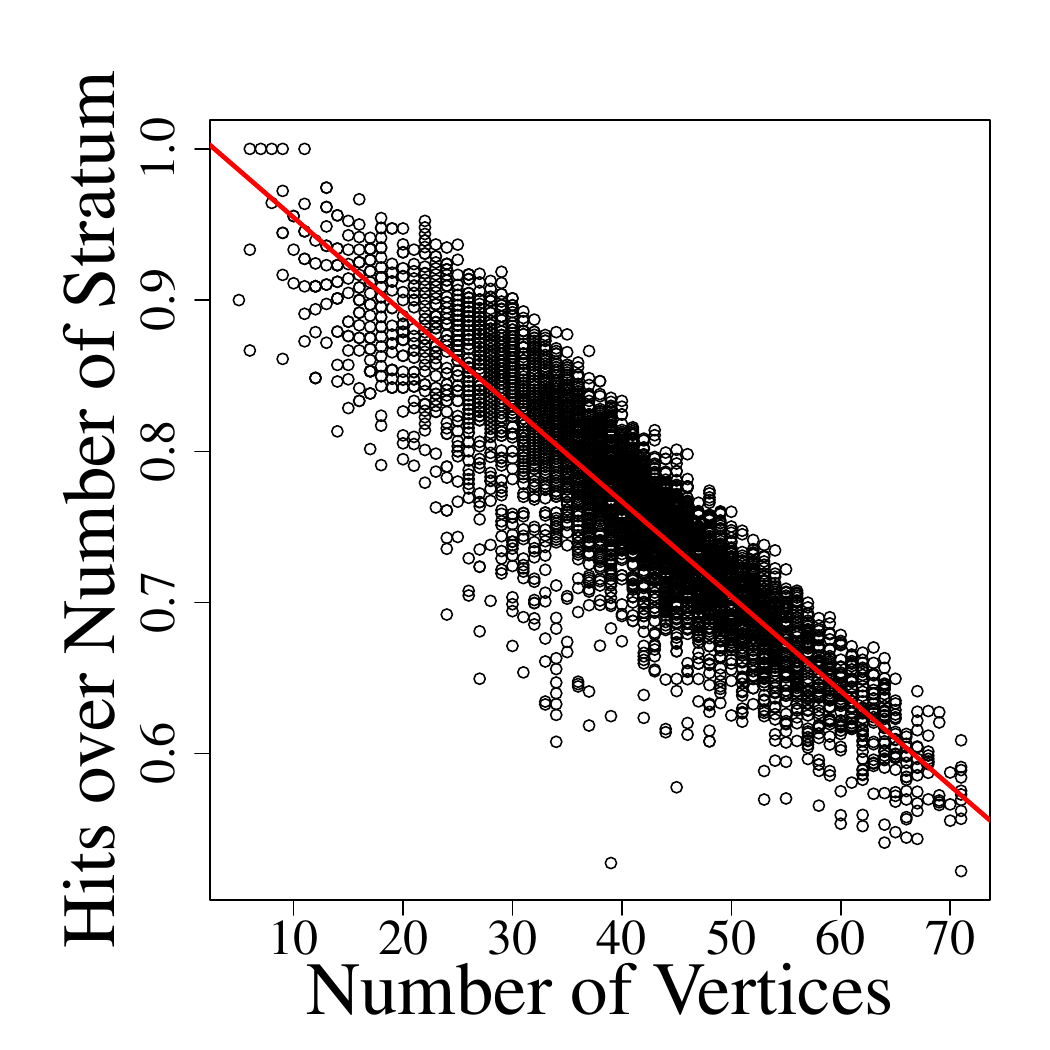}
        \caption{\emnistI}
        \label{subfig:uniform-mnist-001-approx}
    \end{subfigure}
    \begin{subfigure}[b]{0.32\textwidth}
        \includegraphics[width=\textwidth]{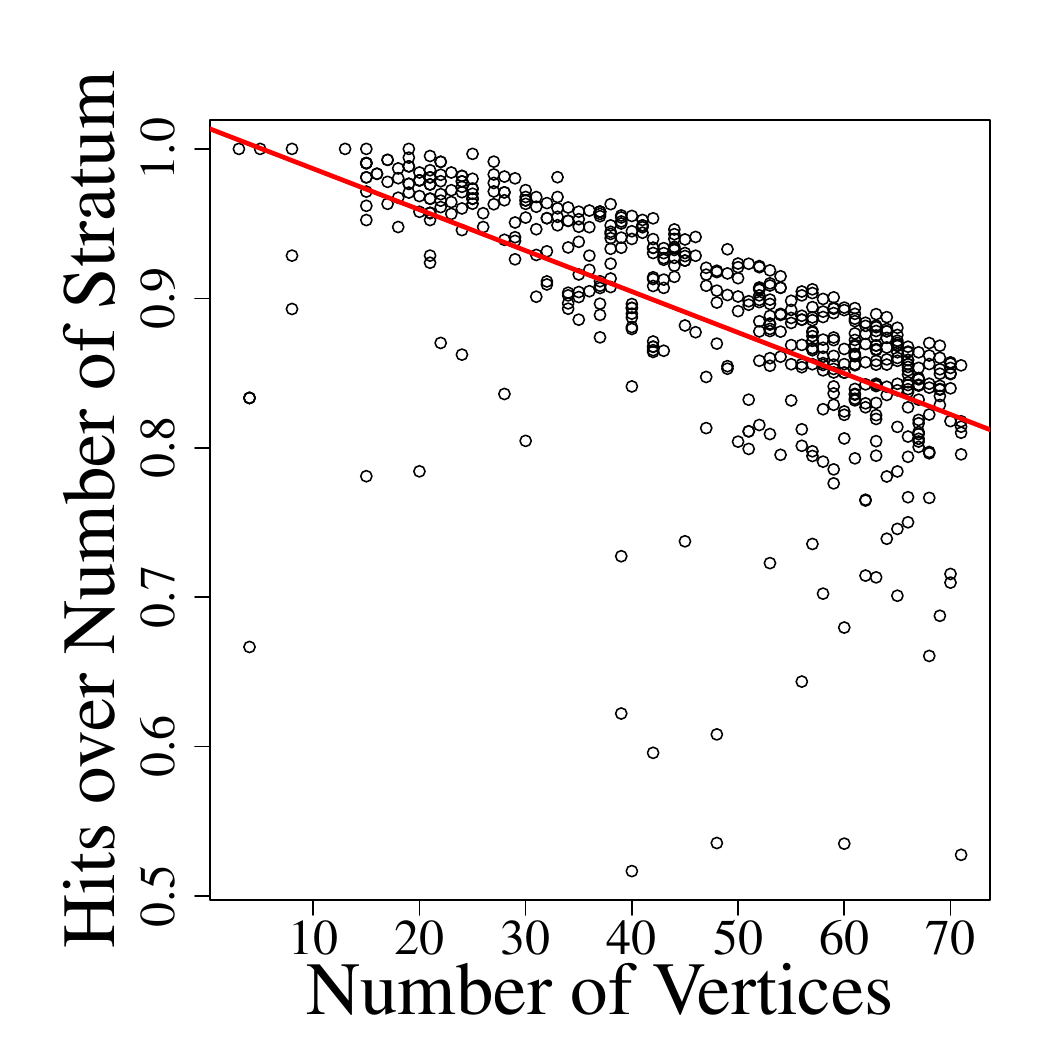}
        \caption{\mpegI}
        \label{subfig:uniform-mpeg7-001-approx}
    \end{subfigure}
    \caption{Plot of the ratio of hit stratum over the total number of
        strata versus the number of
        vertices for \randpts, \emnistI, and \mpegI.
        \label{fig:uniform_sampling-001-approx}
    }
\end{figure*}

We select 16384 uniformly distributed directions $\Delta$ from $\sph^1$.
For each graph, we compute the coarse stratification and
compute the proportion of strata hit by $\Delta$.
We observe that in all cases, even for small graphs with~$60$ or more
vertices, we miss more than 10\% of the strata, implying that if we
want to use a small set of descriptors,
we must be willing to miss some strata.

In summary, across all datasets (\randpts, \emnistI, and \mpegI), using uniform sampling
to hit each stratum becomes less effective as the number of vertices increases, with varying
degrees of correlation strength.

\begin{itemize}
    \item \randpts: $\log(m) = 0.956464 - 1.347025 \log(n_0) - 0.314594 \log^2(n_0)$
    \item \emnistI: $\log(m) = 1.017 - 0.006265 \log(n_0)$
    \item \mpegI: $\log(m) = 1.0141388 - 0.0027408 \log(n_0)$
\end{itemize}

\subsection{Small Graph Experiment}
\begin{figure}[h]
    \centering
    \includegraphics[width=.65\columnwidth]{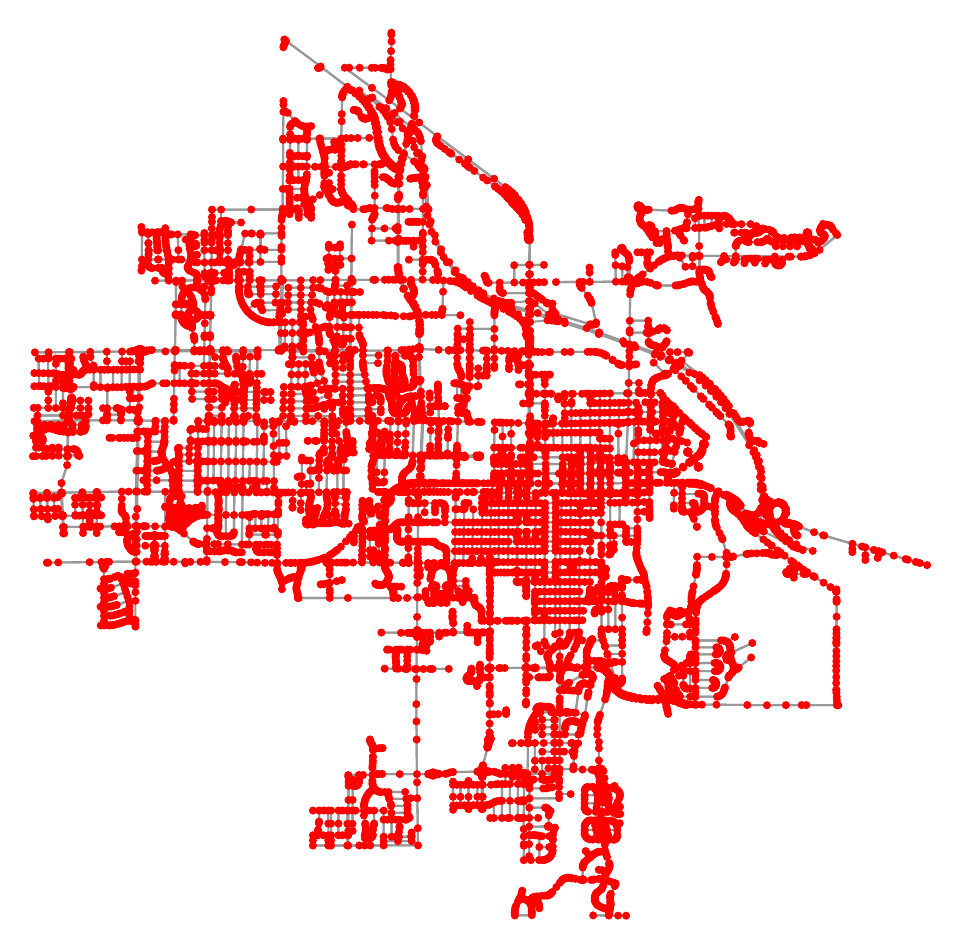}
    \caption{Bozeman, MT OSMnx Street Network}
    \label{fig:bozeman-street-network}
   \end{figure}

   Using too many directions ensures accuracy, but makes computation impractical due to the large number of descriptors.
   On the other hand, selecting directions without considering the $\epsilon$-net for simplicial complexes risks missing important topological features.
   This raises a key question: how can we effectively select directions for discretizing topological descriptors?

   Our goal in this experiment is to demonstrate that, in practice, only a small number of directions are needed to distinguish real-world data. We
   consider plane graphs as our dataset. A
   \emph{plane graph} is a graph embedded in the plane (i.e., a graph $G$ with a
   bijective continuous map $G\to \R^2$).

   \figref{bozeman-street-network} depicts the street network of Bozeman,
   MT, obtained from Python’s OSMnx package \cite{boeing2025osmnx}. We extract graphs from
   this network using bounding boxes of dimensions 60×60, 80×80, and 30×30 to obtain subgraphs with four, five, and six vertices,
   respectively. For each sampled graph, we generate all possible plane
   subgraphs. We restrict our study to graphs with at
   most six vertices because, as the number of vertices increases,
   enumerating all possible plane subgraphs becomes computationally
   intense.

   At the start, we consider a subgraph of the Bozeman road network and generate all possible plane subgraphs from it. We create a coarse stratification associated with the main subgraph. From this stratification, we randomly select a stratum and then a random direction within that stratum, denoted \( d_0 \). Using \( d_0 \), we compute the directional persistence diagram for every plane subgraph. We compare these diagrams to determine whether the original Bozeman subgraph can be uniquely identified from its plane subgraphs. If the identification is not unique—that is, if multiple plane subgraphs yield identical persistence diagrams—we select another random stratum (different from previously chosen ones) and a random direction within it. We then use persistence diagrams from both directions to attempt unique identification. This process continues by sampling additional random directions from different strata until the original subgraph can be uniquely identified from its set of plane subgraphs.

   \begin{figure*}[ht]
     \begin{subfigure}[b]{0.32\textwidth}
         \includegraphics[width=\textwidth]{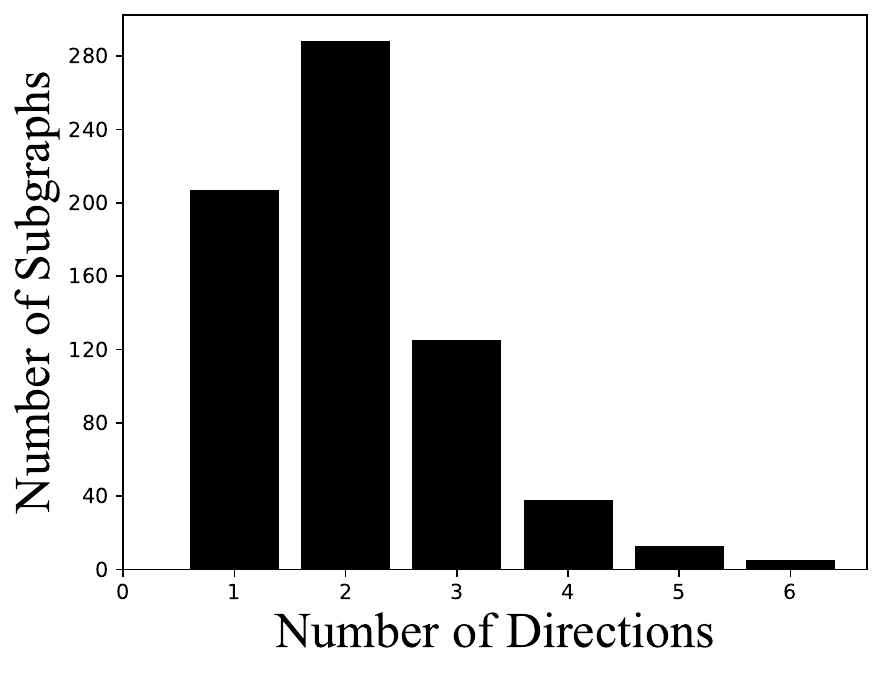}
         \caption{Subgraphs on 4 Vertices}
         \label{subfig:4-vertices}
     \end{subfigure}
     \begin{subfigure}[b]{0.32\textwidth}
         \includegraphics[width=\textwidth]{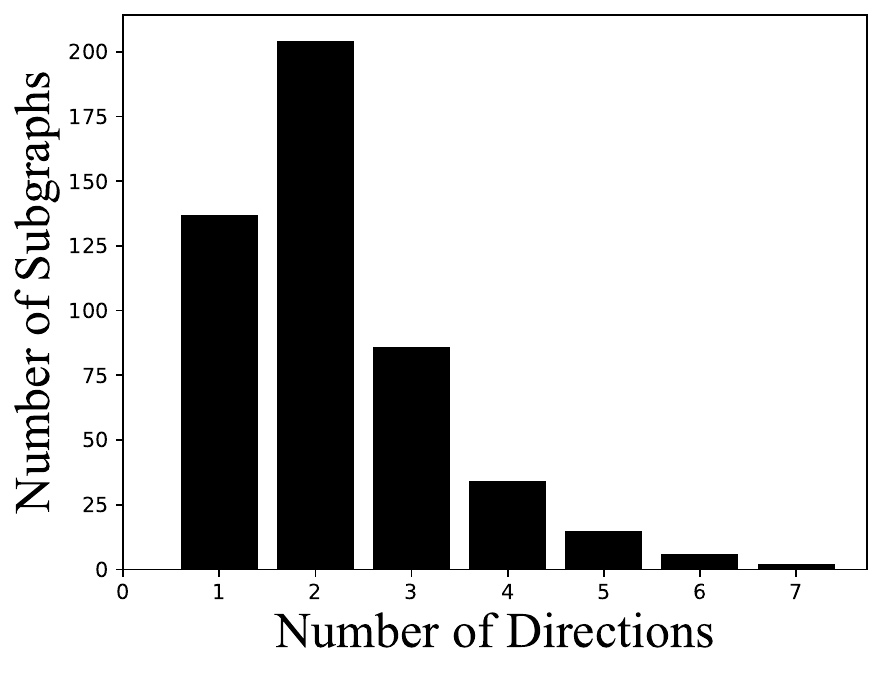}
         \caption{Subgraphs on 5 Vertices}
         \label{subfig:5-vertices}
     \end{subfigure}
     \begin{subfigure}[b]{0.32\textwidth}
         \includegraphics[width=\textwidth]{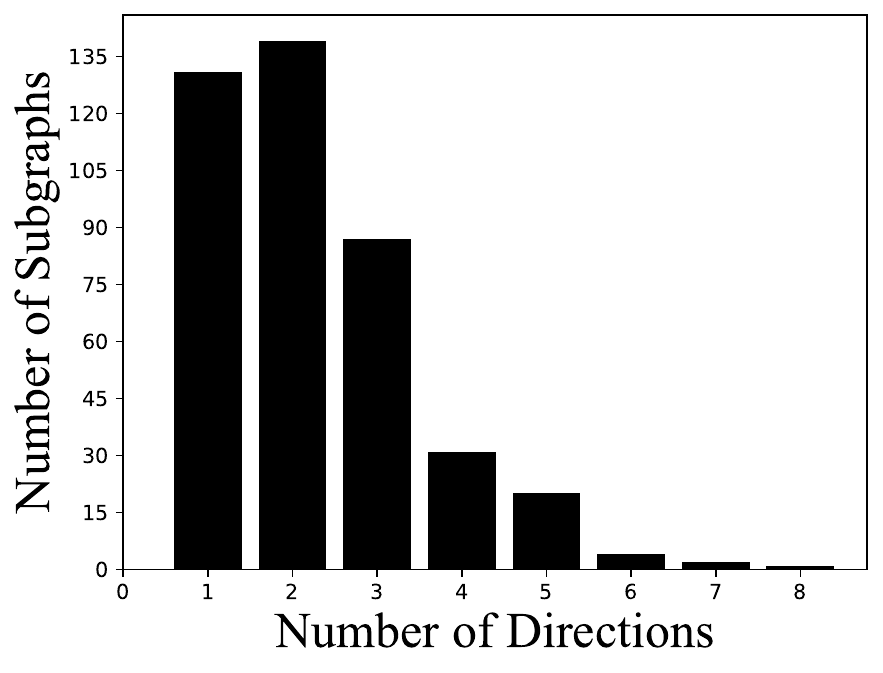}
         \caption{Subgraphs on 6 Vertices}
         \label{subfig:6-vertices}
     \end{subfigure}
     \caption{For each subgraph, we randomly sample from a discretization of $\S^1$ and add equally spaced directions until no other plane graphs exist with the same set of persistence diagrams.}
     \label{fig:small-graph-exp}
   \end{figure*}

   As depicted in \figref{small-graph-exp}, majority
   of the subgraphs sampled required only three directions.
   Generally, we find that between two and four directions
   are sufficient; however, we observed a small number of
   road networks that require a larger number of directions.
   These networks are depicted in \figref{small-graph-exp-directions}.
   As \assref{general} suggests, all the networks that
   require more directions have (near) colinearities,
   depicted as white vertices in the figure.

   \begin{figure*}[ht]
     \begin{subfigure}[b]{0.32\textwidth}
         \includegraphics[width=\textwidth]{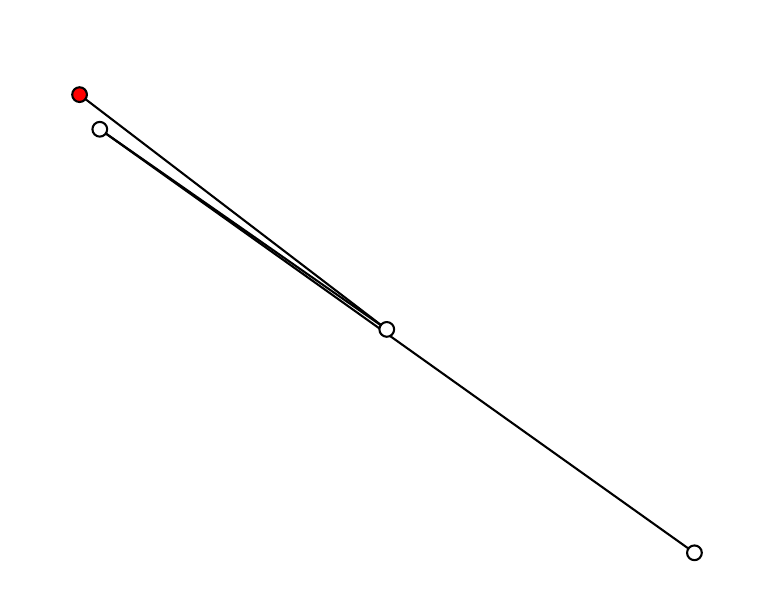}
         \caption{6 directions}
         \label{subfig:4-nodes-example}
     \end{subfigure}
     \begin{subfigure}[b]{0.32\textwidth}
         \includegraphics[width=\textwidth]{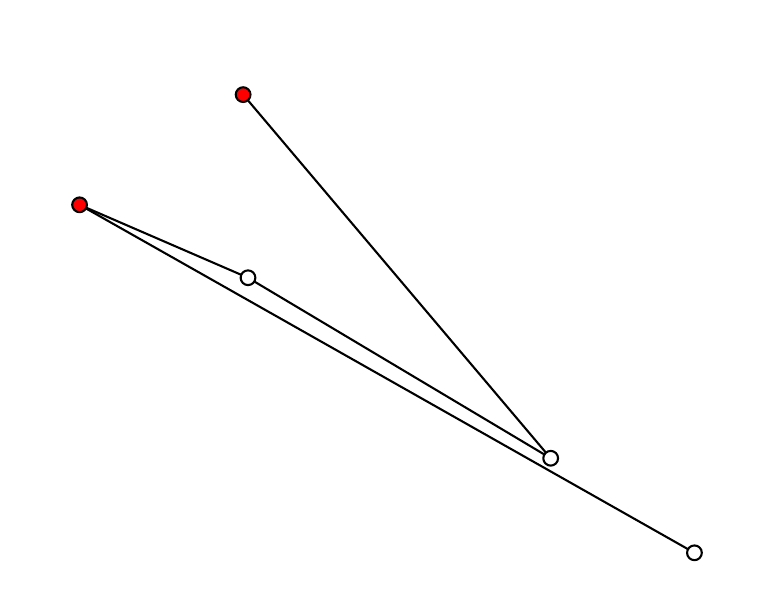}
         \caption{9 directions}
         \label{subfig:5-nodes-example}
     \end{subfigure}
     \begin{subfigure}[b]{0.32\textwidth}
         \includegraphics[width=\textwidth]{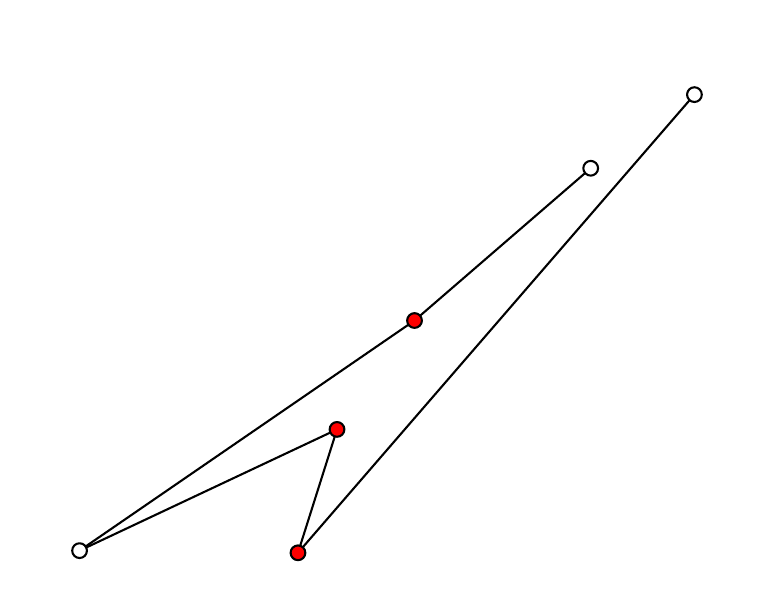}
         \caption{8 directions}
         \label{subfig:6-nodes-example}
     \end{subfigure}
     \caption{Graphs with 4, 5, and 6 nodes that require more than four directions.}
     \label{fig:small-graph-exp-directions}
   \end{figure*}

   The final step involves plotting a histogram depicting
   the number of subgraphs versus the number of directions
   required. The standard distribution of this histogram
   serves as the average number of directions from the coarse
   stratification needed to uphold the faithfulness
   of the descriptor set.

\section{Implications of Too Few}\label{sec:loss}
Our experiments show that using \emph{too many} directions can cause oversampling,
while using \emph{too few} directions often fails to capture the full
complexity of data.  Balancing the two, researchers often
err on the side of too few, for what is the worth of a descriptor set if there
are too many descriptors to realistically compute? Thus, we address one final
question: \emph{how far apart can two simplicial complexes be if their
topological descriptors are the same on a
fixed set of~directions}?

\begin{theorem}[Lost Vertex]\label{thm:loss}
    Let~$\simComp$ be a simplicial complex such that the following conditions
    are met:
    \begin{enumerate}[(a)]
        \item The vertices are in general position (\assumptref{general}),
            \label{cond:lossG-gen}
        \item There exist distinct vertices $u,v,w$ such that
            $||u-v|| = ||w-v||$,
            \label{cond:lossG-eq}
        \item The open ball $\ball{v}{||w-v||} \cap \simComp =
            \{v, (u,v), (v,w)\}$,\label{cond:lossG-ballV}.
    \end{enumerate}

    Let $\Delta \subset \S^d\setminus \obs{v}{\simComp}$.
    Let~$\gscSpace_{\Delta}$ denote the set of geometric simplicial complexes
    in~$\R^d$  such that, for all $\dir \in \Delta$, the corresponding
    filtrations are equivalent:~$\filt{K}{\dir} \cong
    \filt{\sampsim}{\dir}$.
    Let~$\theta= \pi - \angle{uvw}$.
    Then, there exists a simplicial complex~$\sampgraph \in \gscSpace_{\Delta}$ such that:
    \begin{equation}\label{eq:haus-bound}
        d_H(\simComp,\sampsim) = ||v-w||\cos{\frac{\theta}{2}} > 0.
    \end{equation}
\end{theorem}

\begin{proof}
    Let $\sampgraph = \left\{(u,w)\right\} \cup
        \simComp  \setminus \{ v, (u,v), (v,w)\}$.
    For any $\dir \notin \obs{v}{\simComp}$, we know that we always see the three named
    vertices in one of two orders: $u$ before~$v$ before $w$, or $w$ before~$v$
    before $u$.  This property, and by the definition of $\sampgraph$,
    we know that~$\filt{\simComp}{\dir} \cong \filt{\sampgraph}{\dir}$ for all
    $\dir \in \Delta$.
    Thus,~$\sampgraph \in \gscSpace_{\Delta}$.

    Let $\pi(v)$ denote the orthogonal
    projection of~$v$ onto the edge~$(u,w)$.
    By the definition of
    $\simComp$ and by
    \condref{lossG}{ballV}, we know that the Hausdorff distance between
    $\simComp$ and $\sampgraph$ is as follows:
    \begin{equation}
        d_H(\simComp, \sampgraph) = d_H(v,\sampgraph) = d_H(v,(u,w)).
    \end{equation}
    In other words,
    the Hausdorff distance between $\simComp$ and $\sampgraph$ is the Euclidean
    distance between $v$ and $\pi(v)$.

    By trigonometric properties of
    isosceles triangles, we know that~$||v-\pi(v)||=
    ||v-w|| \cos{\frac{\theta}{2}}$.

    Thus, we~conclude:
    \begin{equation}
        d_H(\simComp,\sampgraph) = ||v-w||\cos{\frac{\theta}{2}}.
    \end{equation}

    Finally, because $\theta= \pi - \angle{uvw}$ and by \condref{lossG}{gen},
    we know that $\theta \in (0,\pi)$.  Furthermore, because $v \neq w$ by
    \condref{lossG}{eq}, we obtain $||v-w||\cos{\frac{\theta}{2}} > 0$.
\end{proof}

\section{Discussion}\label{sec:discussion}
The experiments in this paper are all in $\R^2$.  They could, of course, be
extended to $\R^3$ or even $\R^d$, as all of the theory and definitions allow for arbitrary
dimensions. However, in this paper, we sought to find simple, yet common, data
examples that demonstrate how bad the behavior of under- or over-sampling
directions could be. Extending these experiments to $\R^3$ or higher
would add computational complexity, as data structures defined
over~$\S^1$ are relatively simple, but over $\S^2$ require more nuances. The
observed behavior would get worse in higher dimensions
(e.g.
Figure \ref{subfig:6-nodes-example}).

While, in theory, faithful representations of shape need many directions, in
practice, smaller sets of directions seem to suffice.  Of course, loss of information is
expected in data summaries, and the impact of that loss depends on the
particular application and the machine learning or statistical tools used.
For example, in cluster analysis, classification in (small) descriptor space is
much easier than classification in (full) shape space, and approximate
clustering is often acceptable. We hope that this paper provides some insight
into how much information about data is lost
with the smaller sets used, so that, when used in applications,
scientists better understand the limits of these descriptors.

\paragraph{Acknowledgments}
The authors thank Luke Padula and James Wilson for their assistance with
implementing prototypes of the experiments and preprocessing data.
We also thank Lorin Crawford for conversations around the use of topological
transforms in practice.

This work was partially supported by the National Science Foundation grants
DMS~1854336 and CCF~2046730.

\bibliographystyle{plain}
\bibliography{references}

\end{document}